\documentclass[sigconf]{acmart}
\newtheorem{algorithm}{Algorithm}
\newtheorem{remark}{Remark}

\usepackage{hyperref}
\usepackage{xcolor}
\usepackage{tabularx}
\usepackage{booktabs}
\usepackage{subcaption}
\usepackage{algorithm}
\usepackage{algorithmicx}
\usepackage{algpseudocode}

\AtBeginDocument{%
  }

\setcopyright{acmlicensed}
\copyrightyear{2024}
\acmYear{2024}
\acmDOI{}

\acmConference[]{}{}{ }
\acmISBN{}

\begin{document}

\title{Embedding Integer Lattices as Ideals into  Polynomial Rings}

\author{Yihang Cheng, Yansong Feng and Yanbin Pan}
\affiliation{%
  \institution{Key Laboratory of Mathematics Mechanization, Academy of Mathematics and Systems Science, Chinese Academy of Sciences\\
  School of Mathematical Sciences, University of Chinese Academy of Sciences}
  \streetaddress{Zhongguancun East Road 55}
  \city{Beijing 100190}
  \country{China}
}
\email{chengyihang15@mails.ucas.ac.cn, {fengyansong, panyanbin}@amss.ac.cn}



\renewcommand{\shortauthors}{Cheng et al.}

\begin{abstract}
  Many lattice-based crypstosystems employ ideal lattices for high efficiency. However, the additional algebraic structure of ideal lattices usually makes us worry about the security, and it is widely believed that the algebraic structure will help us solve the hard problems in ideal lattices more efficiently. In this paper, we study the additional algebraic structure of ideal lattices further and find that a given ideal lattice in a polynomial ring can be embedded as  an ideal into infinitely many different polynomial rings by the coefficient embedding. We design an algorithm to verify whether a given full-rank lattice in $\mathbb{Z}^n$ is an ideal lattice and output all the polynomial rings that the given lattice can be embedded into as an ideal with time complexity $\mathcal{O}(n^3B(B+\log n)$, where $n$ is the dimension of the lattice and $B$ is the upper bound of the bit length of the entries of the input lattice basis. We would like to point out that Ding and Lindner proposed an algorithm for identifying ideal lattices and outputting a single polynomial ring that the input lattice can be embedded into with time complexity $\mathcal{O}(n^5B^2)$ in 2007. However, we find a flaw in Ding and Lindner's algorithm that causes some ideal lattices can't be identified by their algorithm.
  
\end{abstract}




\maketitle

\section{Introduction}
\subsection{The Development of Ideal Lattices}
The research on lattice-based cryptography was pioneered by Ajtai \cite{ref1} in 1996. He presented a family of one-way function  based on the Short Integer Solution (SIS) problem, which has the average-case hardness under the worst-case assumptions for some lattice problems. In 1997,  Ajtai and Dwork \cite{ref2} introduced a public-key cryptosystem, whose average-case security can be based on the worst-case hardness of the unique-Shortest Vector Problem.  In 2005,  Regev \cite{ref4} proposed another problem with average-case hardness,  the Learning with Errors problem (LWE), and also a  public-key encryption scheme based on LWE. Because of the average-case security, lattice-based cryptography has drawn considerable attentions from then on. 

Although there have been many cryptographic schemes based on LWE and SIS,   the main drawback of such schemes is their limited efficiency, due to its large key size and slow computations. Especially, as  the development of quantum computers, it becomes more urgent to design more practical lattice-based cryptosystems, since lattice-based cryptosystems are widely believed to be quantum-resistant.  
To improve the efficiency, additional algebraic structure is involved in the  lattice to construct more practical schemes. Among them,  ideal lattice plays an important role. 
     
In fact, as early as in 1998, Hoffstein, Pipher, and Silverman \cite{ref3} introduced a lattice-based public-key encryption scheme known as NTRU, whose security is related to the ideal in the ring $\mathbb{Z}[x]/(x^n-1)$. Due to the cyclic structure of the ideal lattice, the efficiency of NTRU is very high. Later, in 2010, Lyubashevsky, Peikert and Regev \cite{ref5} presented a ring-based variant of LWE, called Ring-LWE, whose average-case  hardness is based on worst-case assumptions on ideal lattices. In 2017, Peikert, Regev and Stephens-Davidowitz \cite{ref6} refined the proof of the security of Ring-LWE for more algebraic number field. After the introduction of Ring-LWE, more and more practical cryptosystems based on ideal lattices have be constructed.

     There are two different ways to define ideal lattices. 
     
     One is induced by the coefficient embedding from ring $\mathbb{Z}[x]/f(x)$ into $\mathbb{Z}^n$. NTRU uses coefficient embedding to define its lattice. It is very convenient to implement cryptosystems based on  Ring-LWE with the coefficient embedding. In fact, almost all the ideal lattice-based cryptosystems are implemented via the coefficient embedding. However, it seems not easy to clarify the hardness of  problems for the coefficient-embedding ideal lattice in general.

     The other one is defined by the canonical embedding from the algebraic integer ring of some number field $K$ into $\mathbb{C}^n$. This type of ideal lattice is usually employed in the security proof or hardness reduction in Ring-LWE based cryptography.
     
     It is widely believed that the additional algebraic structure of ideal lattice will help us solve its hard problems  more efficiently. 
     
     In 2016, Cramer, Ducas, Peikert and Regev \cite{ref7} introduced a polynomial-time quantum algorithm to solve $2^{\sqrt{n\text{log}n}}$-SVP in principal ideal lattices in the algebraic integer ring of $\mathbb{Q}(\zeta_m)$, where $m$ is a power of some prime. In 2017, Cramer, Ducas and Wesolowski \cite{ref8} extended the result to general ideals. In the same year,  Holzer, Wunderer and Buchmann \cite{ref9} extended the field to be $\mathbb{Q}(\zeta_m)$, where $m=p^aq^b$ and $p$, $q$ are different primes. 
     
     In 2019, Pellet-Mary, Hanrot and Stehl\'{e} \cite{ref10} introduced a pre-processing method (PHS algorithm) to solve $\gamma$-SVP for ideal lattices in any number field. The pre-processing phasing takes exponential time. Let $n$ be the dimension of the number field $K$ viewed as a $\mathbb{Q}$-vector space. Pellet-Mary \textit{et al.} showed that by performing pre-processing on $K$ in exponential time, their algorithm can, given any ideal lattice $I$ of $O_K$, for any $\alpha \in [0,1/2]$ output a $\exp(\widetilde{O}((n\log n)^{\alpha+1}/n))$ approximation of a shortest none-zero vector of $I$ in time $\exp(\widetilde{O}((n\log n)^{1-2\alpha}/n))+T$.  For the classical method, $T=\exp(\widetilde{O}((n\log n)^{1/2})$ if $K$ is  a cyclotomic field or $T=\exp(\widetilde{O}((n\log n)^{2/3})$ for an arbitrary number field $K$.
     
     In 2020, Bernard and Roux-Langlois \cite{ref19} proposed a new “twisted” version of the PHS  algorithm. They proved that Twisted-PHS algorithm performs at least as well as the original PHS algorithm and their algorithm suggested that much better approximation factors were achieved. In 2022, Bernard,  Lesavourey,  Nguyen and  Roux-Langlois \cite{ref20} extended the experiments of \cite{ref19} to cyclotomic fields of degree up to 210 for most conductors $m$. 
     
 In 2021, Pan, Xu, Wadleigh and Cheng \cite{ref21} found the connection between the complexity of the shortest vector problem (SVP) of prime ideals in number fields and their decomposition groups, and revealed lots of weak instances of  ideal lattices in which SVP can be solved efficiently. In 2022, Boudgoust, Gachon and Pellet-Mary \cite{ref22} generalized the work of Pan \textit{et al.} \cite{ref21} and provided a simple condition under which an ideal lattice defines an easy instance of the shortest vector problem. Namely, they showed that the more automorphisms stabilize the ideal, the easier it was to find a short vector in it.
 
 As mentioned above, almost all the research on SVP is in the canonical-embedding ideal lattices and the research on SVP in the coefficient-embedding ideal lattices is few.  However, in some rings, such as $\mathbb{Z}[X]/(x^n+1)$ where $n=2^k$ for $k\geq 1$, the SVPs induced by the two different embeddings are almost equal. We refer to  \cite{ref11} for more details.

\subsection{Our contribution}
In this paper, our main contribution is to find that an ideal lattice in the ring $\mathbb{Z}[x]/f(x)$ can be embedded into infinitely many rings $\mathbb{Z}[x]/g(x)$ as ideals, where $f(x)$ and $g(x)$ are monic and $f(x)$, $g(x)\in\mathbb{Z}[x]$ (Theorem~\ref{kthm}). Besides, corresponding to our finding, we show an efficient algorithm for computing all the rings that an ideal lattice can be embedded into as ideals and also judging whether a given integer lattice can be embedded as an ideal into a polynomial ring like $\mathbb{Z}[x]/f(x)$ with time complexity $\mathcal{O}(n^3B(B+\log n)$, where $n$ is the dimension of the lattice and $B$ is the upper bound of the bit length of the entries of the input lattice basis (Algorithm~\ref{alg:identify}).

Although, in 2007, Ding and Lindner \cite{ref13} proposed an algorithm for identifying ideal lattice that output a single polynomial ring which the input lattice can be embedded into as an ideal with time complexity $\mathcal{O}(n^5B^2)$, we find that there is a flaw in Ding and Lindner's algorithm. More exactly, some ideal lattices can't be identified by their algorithm and we give a non-trivial example in Section 4. Besides, ignoring the flaw, our algorithm is more efficient and output more polynomial rings than Ding and Lindner's algorithm.

On one hand, our finding reveals that an ideal lattice in $\mathbb{Z}[x]/f(x)$ can be viewed as an ideal lattice in $\mathbb{Z}[x]/g(x)$ for infinitely many different $g(x)$ and it is widely believed that some additional algebraic structures may lead a more efficient algorithm to solve the hard problems in ideal lattice than  general lattices, such as \cite{ref7}, \cite{ref10}.  Hence, we may embed the given ideal lattice into a well-studied ring as an ideal lattice and use the algebraic structure of the well-studied ring to solve the hard lattice problems more efficiently.

On the other hand, we test the proportion of ideal lattices in plain integer lattices by our algorithm and find that the proportion decreases very fast with the increase of the lattice dimension and upper bound of the bit length of the entries of the input lattice basis. Our test data indicates that the ideal lattice is actually very rare.

Finally, we provide an efficient open source implementation of our algorithm for identifying ideal lattices in SageMath. The source code is available at:\begin{center}
  \textcolor{blue}{\url{https://github.com/fffmath/Identifying-Ideal-Lattice}}.
\end{center}
With this implementation, we conducted several experiments, and the experimental results are presented in Appendix~\ref{AppendixA}.

\subsection{Roadmap}The paper is organized as follows. In Section 2, some preliminaries are presented. In Section 3, we show embedding relation between integer lattices and polynomial rings, and the theoretic basis of Algorithm \ref{alg:identify} is also presented. In Section 4, we propose the algorithm for identifying a coefficient-embedding ideal lattice together with the complexity analysis and the comparison to Ding and Lindner's algorithm. The appendix contains our experimental results and reference.

\section{Preliminaries}
\subsection{Notation}

In this paper we denote by $\mathbb{C}$, $\mathbb{R}$, $\mathbb{Q}$ and $\mathbb{Z}$ the complex number field, the real number field, the rational number field and the integer ring respectively.

We denote a matrix by a  capital letter in bold and denote a vector by a lower-case letter in bold. To  represent the element of a matrix, we use the lower-case letter. For example,  the element of matrix $\mathbf{A}$ at the $i$-th row and $j$-th column is denoted by $a_{ij}$, while its $i$-th row is denoted by $\mathbf{a}_i$. Since we have the inner products in $\mathbb{R}^n$ and $\mathbb{C}^n$ respectively, we can define the norm of vectors,  that is, $ \Vert \mathbf{v} \Vert :=<\mathbf{v},\mathbf{v}>$ in $\mathbb{R}^n$ and  $ \Vert \mathbf{v} \Vert :=<\mathbf{v},\overline{\mathbf{v}}>$ in $\mathbb{C}^n$.

For two integers $a$ and $b$, $a| b$ means that $b$ is divisible by $a$. Otherwise, we write $a\not| \ b$. For integer $a$ and a matrix $\mathbf{A}$, $a| \mathbf{A}$ means that every entry of $\mathbf{A}$ can be divisible by $a$.


For a polynomial $f(x)\in\mathbb{Z}[x]$, denote by $\mathbb{Z}[x]/f(x)$ for simplicity the quotient ring $\mathbb{Z}[x]/(f(x)\mathbb{Z}[x])$.

For a map $\sigma$, and a set $S$, denote by $\sigma(S)$ the set $\{\sigma(x):x\in S\}$.

\subsection{Lattice}
Lattices are  discrete subgroups of $ \mathbb{R}^m$, or equivalently,
\begin{definition}{(Lattice)}
Given n linearly independent vectors $\mathbf{B}=\begin{pmatrix}\mathbf{b}_1 \\ \mathbf{b}_2 \\\vdots \\ \mathbf{b}_n\end{pmatrix}$, where $\mathbf{b}_i\in\mathbb{R}^m$, the lattice $\mathcal{L}(\mathbf{B})$ generated by $\mathbf{B}$ is defined as  follows: \[\mathcal{L}(\mathbf{B})=\{\sum_{i=1}^n x_i\mathbf{b}_i : x_i\in\mathbb{Z}\}=\{\mathbf{xB} : \mathbf{x}\in\mathbb{Z}^n\}.\] 	
\end{definition}
We call $\mathbf{B}$ a basis of $\mathcal{L}(\mathbf{B})$, $m$ and $n$ the dimension and the rank of $\mathcal{L}(\mathbf{B})$ respectively. When $m=n$, we say $\mathcal{L}(\mathbf{B})$ is full-rank.

When $n>1$, there are infinitely many bases for a lattice $\mathbf{\mathcal{L}}$, and any two bases are related to each other by a unimodular matrix, which is an invertible integer matrix. More precisely, given a lattice $\mathcal{L}(\mathbf{B}_1)$, $\mathbf{B}_2$ is also a basis of the lattice if and only if there exists a unimodular matrix $\mathbf{U}$ s.t. $\mathbf{B}_1=\mathbf{U} \mathbf{B}_2$.

\paragraph{Hard problems in lattices}

The shortest vector problem (SVP)  is one  of the most famous hard problems in lattices.

SVP is the  question of finding a nonzero shortest  vector in a given lattice $\mathcal{L}$, whose length is denoted by $\lambda_1(\mathcal{L})$. 
The approximating-SVP with factor $\gamma$, denoted by $\gamma$-SVP, asks to find a short nonzero lattice vector $\mathbf{v}$ such that $$\|\mathbf{v}\|\le\gamma\cdot\lambda_1(\mathcal{L}).$$

In fact, The hardness of $\gamma$-SVP depends on $\gamma$. When $\gamma=1$, $\gamma$-SVP is exactly the original SVP, and for constant $\gamma$, this problem is known to be NP-hard under randomized reduction \cite{ajtai98}. 
Many cryptosystems are based on the hardness of (decision) $\gamma$-SVP when $\gamma$ is in polynomial size. By now we have not found any  polynomial-time classical algorithm to deal with such cases. The existing polynomial algorithms such as LLL \cite{ref14}, BKZ \cite{ref15} can find the situation when $\gamma =\text{exp}(n)$.

 \subsection{Hermite Normal Form And Smith Normal Form}
For the integer matrix, there is a very important standard form known as the Hermite Normal Form (HNF). For simplicity, we just present the definition of HNF for the non-singular integer matrix.
\begin{definition}{(Hermite Normal Form)} A non-singular matrix $\mathbf{H}\in\mathbb{Z}^{n \times n}$ is said to be in HNF, if
	\begin{itemize}
		\item $h_{i,i} > 0$ for $1\leq i\leq n.$
		\item $h_{j,i} =0 $ for $1\leq j < i \leq n.$
		\item $0 \leq h_{j,i}< h_{i,i}$ for $1\leq i< j\leq n.$
	\end{itemize}
\end{definition}

The Hermite Normal Form has some important properties. See \cite{DKT87,MW01,LP19} for more details. 
\begin{lemma}
For any  integer matrix $\mathbf{A}$, there exists a unimodular matrix $\mathbf{U}$ such that $\mathbf{H}$=$\mathbf{U} \mathbf{A}$ is in HNF. Moreover, HNF can be computed in polynomial time.
\end{lemma}

For integer lattices, we have 
\begin{lemma}\label{hnfbasis}
For any lattice $\mathcal{L}\subset \mathbb{Z}^n$, there exists a unique basis $\mathbf{H}$ in HNF. We call $\mathbf{H}$ the HNF basis of $\mathcal{L}$.
\end{lemma}

  Sometimes we do not need the whole HNF of an integer matrix. So we  introduce the Incomplete Hermite Normal Form of an integer matrix, which is also  a special basis of the integer lattice.
  
  \begin{definition}{(Incomplete Hermite Normal Form)} A non-singular matrix $\mathbf{B}\in\mathbb{Z}^{n \times n}$ is said to be in Incomplete Hermite Normal Form, if
  \begin{itemize}
  	\item $b_{n,n}>0$;
  	\item $b_{i,n} = 0 \mbox{ for } 1\leq i\leq n-1.$
  \end{itemize}
  \end{definition}

Given a full-rank integer matrix  $\mathbf{B}$,
\[ \mathbf{B}= \begin{pmatrix} b_{1,1}&b_{1,2}&\cdots&b_{1,n} \\ b_{2,1}&b_{2,2}&\cdots&b_{2,n} \\ \vdots&\vdots&\ddots&\vdots \\b_{n,1}&b_{n,2}&\cdots&b_{n,n} \end{pmatrix},\]
it is well known that by the Extended  Euclidean Algorithm we can find a unimodular matrix $\mathbf{U}$, such that 
\[ \mathbf{U} \begin{pmatrix} b_{1,n} \\ b_{2,n} \\ \vdots \\b_{n,n} \end{pmatrix} = \begin{pmatrix}0 \\ 0 \\ \vdots \\d \end{pmatrix},\]
where $d = \gcd( b_{1,n}, b_{2,n},...,b_{n,n})$. Then we have 
$$ \mathbf{B}'= \mathbf{U} \mathbf{B} =  \begin{pmatrix} \mathbf{D}&\mathbf{0} \\ \mathbf{b'}& d \end{pmatrix} $$
is in Incomplete Hermite Normal Form, where  $\mathbf{D} \in \mathbb{Z}^{(n-1) \times (n-1)}$, $\mathbf{b}' \in \mathbb{Z}^{n-1}$. 

About the Incomplete Hermite Normal Form, it is easy to conclude the following lemma. So we omit the proof.

\begin{lemma} \label{ihnflemma}
	For any non-singular matrix $\mathbf{B}\in\mathbb{Z}^{n \times n}$, the following properties are satisfied:
	\begin{itemize}
		\item we can find a unimodular matrix $\mathbf{U}$ in polynomial time, such that $ \mathbf{B}'= \mathbf{U} \mathbf{B}$ is in Incomplete Hermite Normal Form.
		\item For any unimodular matrix $\mathbf{U}$ and $\mathbf{V}$ such that $ \mathbf{B}'= \mathbf{U} \mathbf{B}$ and $ \mathbf{B}''= \mathbf{V} \mathbf{B}$ both in Incomplete Hermite Normal Form,  $ \mathbf{B}'$ and $ \mathbf{B}''$ are not necessarily equal, but $$b'_{n,n} = b''_{n,n} = \gcd( b_{1,n}, b_{2,n},...,b_{n,n}).$$  Specially, notice that the HNF  $\mathbf{H}$ of $\mathbf{B}$ is also in  Incomplete Hermite Normal Form. We immediately have
		$$ h_{n,n} = \gcd( b_{1,n}, b_{2,n},...,b_{n,n}).$$
	\end{itemize}
\end{lemma}

\begin{definition}{(Smith Normal form)}Let $\mathbf{A}$ be nonzero $m\times n$ matrix over a principal ideal domain $R$, there exist invertible $m\times m$ and $n\times n$-matrices $\mathbf{P},\mathbf{T}$ (with coefficients in $R$) such that the product \[\mathbf{S}=\mathbf{PAT}=\begin{pmatrix}\alpha_1&0&0&\cdots&0\\0&\alpha_2&0&\cdots&0\\0&0&\ddots&\ddots&0\\\vdots&\vdots&\ddots&\alpha_r&\vdots\\\vdots&\vdots&0&0&\vdots\\0&\cdots&\cdots&\cdots& 0\end{pmatrix}\] And the diagonal elements satisfy $\alpha_i|\alpha_{i+1}$ for all $1\leq i < r$. $\mathbf{S}$ is the Smith Normal Form of $\mathbf{A}$, and the elements $\alpha _{i}$ are unique up to multiplication by a unit in $R$ and are called the elementary divisors, invariants, or invariant factors.
\end{definition}

\begin{definition}{(Smith Massager)}Let $\mathbf{A}\in \mathbb{Z}^{n\times n}$ be a non-singular integer matrix with Smith Normal Form $\mathbf{S}$. A matrix $\mathbf{M}\in\mathbb{Z}^{n\times n}$ is a Smith Massager for $\mathbf{A}$ if 

(i) it satisfies that $\mathbf{AM}\equiv 0$ cmod $\mathbf{S}$, and

(ii) there exists a matrix $\mathbf{W}\in \mathbb{Z}^{n\times n} $ such that $\mathbf{WM}\equiv \mathbf{I_n}$ \text{cmod} $\mathbf{S}$.
\end{definition}
\begin{definition}{(cmod)}
Given $\mathbf{B}\in\mathbb{Z}^{m\times n}$ and $\mathbf{S}\in\mathbb{Z}^{n\times n}$, where \[\mathbf{B}=\begin{pmatrix}
    \mathbf{b_1}&\mathbf{b_2}&\cdots&\mathbf{b_n}
\end{pmatrix}\]
\[\mathbf{S}=\begin{pmatrix}
    s_1&0&\cdots&0\\0&s_2&\cdots&0\\\vdots&\vdots&\ddots&\vdots\\0&0&\cdots&s_n
\end{pmatrix}\]
$\mathbf{b_i}$ is the $i$-th column of $\mathbf{B}$ and $\mathbf{S}$ is a diagonal matrix. 

$\mathbf{B}$ cmod $\mathbf{S}$ :=  $\begin{pmatrix}
    \mathbf{b_1}\mod{s_1}&\mathbf{b_2}\mod{s_2}&\cdots&\mathbf{b_n}\mod{s_n}
\end{pmatrix}$
\end{definition}

The definitions of Smith Normal form and Smith Massager will only be used in Theorem~\ref{LatMemthm}, Section~\ref{section:Identifying an Ideal Lattice}.
\subsection{Ideal lattices}

 An algebraic number field $K$ is an extension field of the rationals $\mathbb{Q}$ such that its dimension $ [K : \mathbb{Q}]$ as a $\mathbb{Q}$-vector space (i.e., its degree) is finite.

 An element $x$ in the algebraic number field $K$ is said to be integral over $\mathbb{Z}$ if the coefficients of the minimal polynomial of $x$ over $\mathbb{Q}$ are all integers. All the elements which are integral over $\mathbb{Z}$ in $K$ make up a set  $O_K$. $O_K$ is actually a ring called the algebraic integer ring of $K$ over $\mathbb{Z}$.

  $O_K$ is a finitely generated $\mathbb{Z}$-module of dimension $[K :\mathbb{Q}]$. The basis of $O_K$ as a $\mathbb{Z}$-module is called the integer basis, which is also a basis of $K$ as a $\mathbb{Q}$-vector space.
 
\paragraph{Canonical-embedding ideal lattice}

 If $\Omega\supset K$ is an extension field such that $\Omega$ is algebraically closed over $\mathbb{Q}$, then there are exactly $[K :\mathbb{Q}]$ field embeddings of $K$ into $\Omega$. For convenience, we regard $\Omega$ as the complex field $\mathbb{C}$.

Any ideal of $O_K$ is a full-rank submodule of $O_K$. Let $[K :\mathbb{Q}]=n$.  This structure induces a canonical embedding:
\begin{align*}
\Sigma: O_K&\rightarrow\mathbb{C}^n \\ a&\mapsto (\Sigma_i(a))_{i=1,...,n},
\end{align*}
where $\Sigma_i$'s are the $n$ different embeddings from $K$ into $\mathbb{C}$.

\begin{definition}{(Canonical-embedding Ideal Lattice)}
Given a number field $K$ and  any ideal I of $O_K$, $\Sigma$(I) is called the canonical-embedding ideal lattice.
\end{definition}

\paragraph{Coefficient-embedding ideal lattice} Denote by $\mathbb{Z}^{(n)}[x]$  the set of all the  polynomials in $\mathbb{Z}[x]$ with degree $\leq$$n-1$. We use the symbol $\sigma$ to represent the following linear map: 
\begin{align*} \sigma : \mathbb{Z}^{(n)}[x]&\rightarrow\mathbb{Z}^n \\ \sum_{i=1}^{n} a_ix^{i-1} &\mapsto (a_1,a_1,...,a_n), \end{align*}
where linear map means that 
\begin{itemize}
	\item For any $f(x)$, $g(x)\in \mathbb{Z}^{(n)}[x]$, $\sigma(f(x)+g(x)) = \sigma(f(x))+\sigma(g(x));$
	\item For any $f(x)\in \mathbb{Z}^{(n)}[x]$ and $z\in\mathbb{Z}$, $\sigma(zf(x)) = z\sigma(f(x)).$
\end{itemize}

We can also define its inverse, which is linear too:
\begin{align*} \sigma^{-1} : \mathbb{Z}^n &\rightarrow\mathbb{Z}^{(n)}[x]\\ (a_1,a_1,\cdots,a_n)&\mapsto \sum_{i=1}^{n} a_ix^{i-1}. \end{align*}

In what follows, we focus on ideal lattices induced by ideals of the ring $\mathbb{Z}[x]/f(x)$, where $f(x)$ is a monic polynomial of degree $n$.  Any element in $\mathbb{Z}^{(n)}[x]$ can be viewed as a  representative in the ring $\mathbb{Z}[x]/f(x)$ with $\text{degree}(f(x))\geq n$ \cite{ref13}.
So we abuse the symbol $\sigma$ to represent the  the following coefficient embedding. 
\begin{align*} \sigma : \mathbb{Z}[x]/f(x)&\rightarrow\mathbb{Z}^n \\ \sum_{i=1}^{n} a_ix^{i-1} &\mapsto (a_1,a_2,...,a_n). \end{align*} 

Therefore, under the coefficient embedding, any ideal of $\mathbb{Z}[x]/f(x)$ can be viewed as an integer lattice.

\begin{definition}{(Coefficient-embedding Ideal Lattice)}
Given $\mathbb{Z}[x]/f(x)$, where  $f(x)$ is a monic polynomial of degree n, and any ideal $I$ of $\mathbb{Z}[x]/f(x)$, $\sigma$(I) is called the coefficient-embedding ideal lattice, which is of course an integer lattice.
\end{definition}


Roughly speaking, due to the abundant algebraic structures of the corresponding algebraic integer domains, the hard lattice problems in canonical-embedding ideal lattices are easier to analyse than that in coefficient-embedding ideal lattices. However, as we've introduced in the introduction, in some cases, the results in canonical-embedding ideal lattices can be converted to the results in the coefficient-embedding ideal lattices with small loss.

The following is an important property of ideal lattices, it was proposed by  Zhang, Liu and Lin \cite{ref16}. We present their proof detail for readers to check conveniently.

\begin{lemma}[\cite{ref16}] \label{klemma}
Let $\mathbf{H}$ be the HNF basis of the full-rank coefficient-embedding ideal lattice $\mathcal{L}(\mathbf{B})$ in the ring $\mathbb{Z}[x]/f(x)$.
 \[ \mathbf{H}= \begin{pmatrix} h_{1,1}&0&\cdots&0 \\ h_{2,1}&h_{2,2}&\cdots&0 \\ \vdots&\vdots&\ddots&\vdots \\h_{n,1}&\cdots&\cdots&h_{n,n} \end{pmatrix}.\] Then $h_{i,i}|h_{j,l}$, for ${1\leq l \leq j \leq i \leq n}$. Specially,
 $h_{n,n} \vert h_{i,j}$, ${i,j\leq n}$.
\end{lemma}
\begin{proof} 
By induction on $i$,  it's trivial for $i=1$.
         
         Assume the result holds for $i \leq k \leq n-1$. It remains to show that for $i=k+1$, $h_{k+1,k+1}|h_{j,l}$ where ${1\leq l \leq j \leq k+1 \leq n}$.
          
          Let $\mathbf{h}_i$ be the $i$-th row of $\mathbf{H}$. Note that for any ideal $I$ of $\mathbb{Z}[x]/f(x)$ and for all $g(x)\in I $, $xg(x) \in I$. Specially $x\sigma^{-1}(\mathbf{h}_k) \in I$, where $\sigma$ is the coefficient-embedding. Since  $\mathbf{H}$ is a basis of the ideal lattice, it is very simple to imply that there must exist $y_i \in \mathbb{Z}$, for $i=1,2,\cdots,k+1$ such that:\[ \begin{pmatrix} 0&h_{k,1}&\cdots&h_{k,k}&0&\cdots&0 \end{pmatrix}=\sum_{i=1}^{k+1}y_i\mathbf{h}_i.\]
          
         Hence, \begin{align*}h_{k,k}&=y_{k+1}h_{k+1,k+1} \\ h_{k,k-1}&=y_kh_{k,k}+y_{k+1}h_{k+1,k}\\\vdots\\h_{k,1}&=\sum_{i=2}^{k+1}y_ih_{i,2}\\0&=\sum_{i=1}^{k+1}y_ih_{i,1} \end{align*}

         From the first equation, we get $y_{k+1}=\frac{h_{k,k}}{h_{k+1,k+1}}\in \mathbb{Z}$, and \begin{align*} h_{k+1,k}&=\frac{h_{k,k-1}-y_kh_{k,k}}{h_{k,k}}h_{k+1,k+1}\\h_{k+1,k-1}&=\frac{h_{k,k-2}-y_{k-1}h_{k-1,k-1}-y_kh_{k,k-1}}{h_{k,k}}h_{k+1,k+1}\\\vdots\\h_{k+1,2}&=\frac{h_{k,1}-\sum_{i=2}^ky_ih_{i,2}}{h_{k,k}}h_{k+1,k+1}\\h_{k+1,1}&=\frac{-\sum_{i=1}^ky_ih_{i,1}}{h_{k,k}}h_{k+1,k+1} \end{align*}
          
          From the induction hypothesis, we have $h_{k,k}|h_{j,l}$ for $1\leq l \leq j \leq k \leq n$. So the coefficient of $h_{k+1,k+1}$ in each equation is in fact an integer. Therefore, $h_{k+1,k+1}|h_{k+1,l},1 \leq l \leq k+1$. Since $h_{k+1,k+1}|h_{k,k}$, we know $h_{k+1,k+1}|h_{j,l}$, where $1 \leq l \leq j \leq k+1 \leq n$. Thus, the result holds for $i=k+1$. 
          
          By induction,  $h_{i,i}|h_{j,l}$, ${1\leq l \leq j \leq i \leq n}$. So $h_{n,n}|h_{i,j}$, ${1\leq i \leq j \leq n}$.  Lemma \ref{klemma} follows.
\end{proof}


\subsection{Overview}In the third section, we first show and prove a naturally equivalent definition (Lemma \ref{klemma4ideal}) of integer lattices. It's a direct application of the definition of the coefficient-embedding ideal lattice. Though the result of Lemma \ref{klemma4ideal} may have been used in some earlier research, we haven't found a detailed description. Hence, we rewrite and prove Lemma \ref{klemma4ideal} formally.

Inspired by Lemma \ref{klemma} proposed by Zhang, Liu and Lin \cite{ref16}, we propose Theorem \ref{thmidentify}, another equivalent definition of ideal lattices in Section 3.2. Using this equivalent definition, we design Algorithm \ref{alg:identify} to verify whether an integer lattice is an ideal lattice.

In Section 3.3, Theorem \ref{kthm} shows that a coefficient-embedding ideal lattice can be embedded into another polynomial ring denoted by $R$ as an ideal of $R$, and for a fixed coefficient-embedding ideal lattice the number of such $R$ is infinite. The proof is also motivated by Lemma \ref{klemma}. Theorem \ref{kthm} guarantees that Algorithm \ref{alg:identify} can output all the polynomial rings which the input integer lattice can be embedded into as ideals.



In the fourth section, we propose Algorithm \ref{alg:identify} to judge whether an integer lattice can be embedded into a polynomial ring as ideals and compute all the rings that the lattice can be embedded into as an ideal if the given lattice is a coefficient-embedding ideal lattice. We analysis the time complexity of Algorithm \ref{alg:identify} and also compare our algorithm to related work.

Finally, we give a brief conclusion. Out experimental data is presented in the Appendix~\ref{AppendixA}.


\section{An ideal lattice can be embedded into different rings}
We stress that in the following, we focus on the coefficient-embedding ideal lattice, and in this section, we'll show how an coefficient-embedding ideal lattice can be embedded into different rings. 


\subsection{Deciding an ideal lattice} We next present an easy way to tell if a given lattice is a  coefficient-embedding ideal lattice in $\mathbb{Z}[x]/f(x)$ or not.

\begin{lemma} \label{klemma4ideal}
	For any monic polynomial $f(x)\in\mathbb{Z}[x]$ with degree $n$, a lattice $\mathcal{L}(\mathbf{B})$ with any basis $\mathbf{B}$ is a  coefficient-embedding ideal lattice in $\mathbb{Z}[x]/f(x)$  if and only if $\sigma(x\sigma^{-1}(\mathbf{b}_i)\mod f(x))\in \mathcal{L}(\mathbf{B})$ for $i=1,\cdots,n$, where $\mathbf{b}_i$ is the $i$-th row vector of $\mathbf{B}$, and $\sigma$ is the map defined in Section 2.3.
\end{lemma}
\begin{proof} 
	If  $\mathcal{L}(\mathbf{B})$ is a  coefficient-embedding ideal lattice in $\mathbb{Z}[x]/f(x)$, then $\sigma^{-1}(\mathbf{b}_i)$'s are in the corresponding ideal. It is obvious that  $x\sigma^{-1}(\mathbf{b}_i)\mod f(x)$ must be in the ideal too, which means that $\sigma(x\sigma^{-1}(\mathbf{b}_i)\mod f(x))\in \mathcal{L}(\mathbf{B})$.

	If there exists a  monic polynomial $f(x)\in\mathbb{Z}[x]$ with degree $n$, such that $\sigma(x\sigma^{-1}(\mathbf{b}_i)\mod f(x))\in \mathcal{L}(\mathbf{B})$ for $i=1,\cdots,n$, we show that $\sigma^{-1}(\mathcal{L}(\mathbf{B}))$ must be an ideal in $\mathbb{Z}[x]/f(x)$. It is easy to check that   $\sigma^{-1}(\mathcal{L}(\mathbf{B}))$  is an additive group, due to the fact that $\sigma$ is an additive homomorphism. Since  $\sigma(x\sigma^{-1}(\mathbf{b}_i)\mod f(x))\in \mathcal{L}(\mathbf{B})$, then for any lattice vector $\mathbf{v} = \sum_{i=1}^n z_i\mathbf{b}_i$, $z_i\in\mathbb{Z}$, we have   $$\sigma(x\sigma^{-1}(\mathbf{v})\mod f(x))=\sum_{i=1}^n z_i \sigma( x \sigma^{-1}(\mathbf{b}_i)\mod f(x)) \in \mathcal{L}(\mathbf{B}). $$
	Applying the result on the lattice vector $\sigma(x\sigma^{-1}(\mathbf{v})\mod f(x))$, we will have
	$$\sigma(x^2\sigma^{-1}(\mathbf{v}))= \sigma(x \sigma^{-1}(\sigma(x\mathbf{v}\mod f(x)))) \in \mathcal{L}(\mathbf{B}). $$
	Hence, for any positive integer $k$, we know that 
		$$\sigma(x^k\sigma^{-1}(\mathbf{v})) \in \mathcal{L}(\mathbf{B}). $$
	Then for any $g(x)=\sum_{i=1}^n g_ix^{i-1}\in\mathbb{Z}[x]/f(x)$ and any lattice vector $\mathbf{v}$, 
		$$\sigma(g(x)\sigma^{-1}(\mathbf{v})\mod f(x))=\sum_{i=1}^n g_i \sigma( x^{i-1} \sigma^{-1}(\mathbf{v})\mod f(x)) \in \mathcal{L}(\mathbf{B}). $$
	The lemma follows.

\end{proof}

\subsection{Equivalent condition}



Inspired by Lemma \ref{klemma}, we find a new equivalent definition of coefficient-embedding ideal lattices.    
  
\begin{theorem}\label{thmidentify}
	Given a full-rank integer lattice $\mathcal{L}(\mathbf{B})$, let $\mathbf{B}'= \begin{pmatrix} \mathbf{D}&\mathbf{0} \\ \mathbf{b'}&b_{n,n}' \end{pmatrix}$ be any  Incomplete Hermit Normal Form of \ $\mathbf{B}$. Then $\mathcal{L}(\mathbf{B})$ is an ideal lattice if and only if there exists a $\mathbf{T} \in \mathbb{Z}^{(n-1) \times n}$, s.t.$\begin{pmatrix} \mathbf{0}&\mathbf{D}\end{pmatrix}=\mathbf{T}\mathbf{B}$. Specially, if $\mathcal{L}(\mathbf{B})$ is an ideal lattice, then taking any $g(x)=x^n+g_nx^{n-1}+\cdots+g_1$ with $\begin{pmatrix}g_1&g_2&\cdots&g_n\end{pmatrix} \in\frac{1}{b_{n,n}'}(\begin{pmatrix}0&\mathbf{b}'\end{pmatrix} +\mathcal{L}(\mathbf{B}))$, $\mathcal{L}(\mathbf{B})$ is also an ideal lattice in the ring $\mathbb{Z}[X]/g(x)$.
\end{theorem}
\begin{proof}
It can be easily check the ``only if'' part by Lemma \ref{klemma4ideal}, since for an ideal lattice $\mathcal{L}(\mathbf{B})$ in $\mathbb{Z}[x]/g(x)$, there exists a $\mathbf{T} \in \mathbb{Z}^{(n-1) \times n}$, s.t. $\begin{pmatrix} \mathbf{0}&\mathbf{D}\end{pmatrix}=\mathbf{T}\mathbf{B}$ if and only if  $\sigma(x\sigma^{-1}(\mathbf{b}'_i)\mod g(x))\in \mathcal{L}(\mathbf{B})$ for $i=1,\cdots,n-1$.

For ``if'' part,
to indicate that $\mathcal{L}(\mathbf{B})$ is an ideal lattice, we need to find a monic polynomial $g(x)$ of degree $n$, s.t. $\mathcal{L}(\mathbf{B})$ can be embedded as an ideal into $\mathbb{Z}[x]/g(x)$, or $\sigma(x\sigma^{-1}(\mathbf{b}'_i)\mod g(x))\in \mathcal{L}(\mathbf{B})$ for $i=1,\cdots,n$  by Lemma \ref{klemma4ideal}.

Note that for any polynomial $g(x)$ with degree $n$, $\sigma(x\sigma^{-1}(\mathbf{b}'_i)\mod g(x))\in \mathcal{L}(\mathbf{B})$ for $i=1,\cdots,n-1$ since there exists a $\mathbf{T} \in \mathbb{Z}^{(n-1) \times n}$, s.t. $\begin{pmatrix} \mathbf{0}&\mathbf{D}\end{pmatrix}=\mathbf{T}\mathbf{B}$. 

It remains to show that there exists a monic polynomial $g(x)$ of degree $n$, such that $\sigma(x\sigma^{-1}(\mathbf{b}'_n)\mod g(x))\in \mathcal{L}(\mathbf{B})$.

We first present a lemma, which will be proven later.
\begin{lemma} \label{identify}	
	If $\begin{pmatrix} \mathbf{0}&\mathbf{D}\end{pmatrix}=\mathbf{T}\mathbf{B}$, then $ \mathbf{B}'/b_{n,n}' \in \mathbb{Z}^{n \times n}$
\end{lemma}

 By Lemma \ref{identify},  $\frac{1}{b_{n,n}'}(\begin{pmatrix}0&\mathbf{b}'\end{pmatrix} +\mathcal{L}(\mathbf{B})) \subset \mathbb{Z}^n$. Taking any
 \begin{equation} \label{equg}
  \mathbf{g}=\begin{pmatrix}g_1&g_2&\cdots&g_n\end{pmatrix} \in\frac{1}{b_{n,n}'}(\begin{pmatrix}0&\mathbf{b}'\end{pmatrix} +\mathcal{L}(\mathbf{B})),
 \end{equation}
   the integer polynomial $g(x)=x^n+g_nx^{n-1}+\cdots+g_1$ is what we want, since
$$
 \sigma(x\sigma^{-1}(\mathbf{b}'_n)\mod g(x)) = \begin{pmatrix}0&\mathbf{b}'\end{pmatrix} - {b_{n,n}'}\begin{pmatrix}g_1&g_2&\cdots&g_n\end{pmatrix}\in \mathcal{L}(\mathbf{B}). $$


It remains to prove Lemma \ref{identify}. Actually, the proof is exactly the same with Lemma \ref{klemma}
 
	


\end{proof}

\subsection{An ideal lattice can be embedded into infinitely many different polynomial rings as ideals}
Given a full-rank ideal lattice $\mathcal{L}(\mathbf{B})$ together with the Incomplete Hermit Normal Form $\mathbf{B}'= \begin{pmatrix} \mathbf{D}&\mathbf{0} \\ \mathbf{b'}&b_{n,n}' \end{pmatrix}$, Theorem \ref{thmidentify} shows that for any $g(x)=x^n+g_nx^{n-1}+\cdots+g_1$ with $\begin{pmatrix}g_1&g_2&\cdots&g_n\end{pmatrix} \in\frac{1}{b_{n,n}'}(\begin{pmatrix}0&\mathbf{b}'\end{pmatrix} +\mathcal{L}(\mathbf{B}))$,  $\mathcal{L}(\mathbf{B})$ is also an ideal lattice in the ring $\mathbb{Z}[X]/g(x)$. The following theorem proves that only if we take $g(x)$ in this way, $\mathcal{L}(\mathbf{B})$ can be viewed as an ideal lattice in the ring $\mathbb{Z}[X]/g(x)$. In other words, the coset $\frac{1}{b_{n,n}'}(\begin{pmatrix}0&\mathbf{b}'\end{pmatrix} +\mathcal{L}(\mathbf{B}))$ can represent the class of all the polynomial rings which the given ideal lattice $\mathcal{L}(\mathbf{B})$ can be embedded into as ideals.

\begin{theorem}\label{kthm}
For any  full-rank coefficient-embedding ideal lattice $\mathcal{L}(\mathbf{B})$ in the ring $\mathbb{Z}[x]/f(x)$, where $f(x)$ is monic and $\text{deg}(f(x))=n$, there exists infinitely many monic $g(x)\in\mathbb{Z}[x]$ with degree $n$, s.t.  $\mathcal{L}(\mathbf{B})$ is also a coefficient-embedding  ideal lattice in $\mathbb{Z}[x]/g(x)$.

More precisely, 
let $d = \gcd( b_{1,n}, b_{2,n},...,b_{n,n})$. Then $\mathcal{L}(\mathbf{B})$ is also a coefficient-embedding  ideal lattice in $\mathbb{Z}[x]/g(x)$, where $g(x)\in\mathbb{Z}[x]$ is a monic polynomial with degree $n$, if and only if 
$$\sigma(f(x)-g(x)) \in  \mathcal{L}(\frac{\mathbf{B}}{d}),$$
or equivalently,
$$g(x) \in f(x) + \sigma^{-1}(\mathcal{L}(\frac{\mathbf{B}}{d})).$$

\end{theorem}
\begin{proof}

Consider the HNF basis of $\mathcal{L}(\mathbf{B})$, 
 \[ \mathbf{H}= \begin{pmatrix} h_{1,1}&0&\cdots&0 \\ h_{2,1}&h_{2,2}&\cdots&0 \\ \vdots&\vdots&\ddots&\vdots \\h_{n,1}&\cdots&\cdots&h_{n,n} \end{pmatrix}.\]
For convenience, we denote  by $\mathbf{h}_i$ the $i$-th row of \ $\mathbf{H}$, and then $\mathbf{h}_i$ is a vector in $\mathbb{Z}^n$. 

(i) If there is a monic $g(x)\in\mathbb{Z}[x]$ with degree $n$, s.t.  $\mathcal{L}(\mathbf{B})$ is also a coefficient-embedding  ideal lattice in $\mathbb{Z}[x]/g(x)$, we next prove that $\sigma(f(x)-g(x)) \in  \mathcal{L}(\frac{\mathbf{B}}{d})$.  

By Lemma \ref{klemma4ideal}, 
 $\mathcal{L}(\mathbf{H}) = \mathcal{L}(\mathbf{B})$ is  a coefficient-embedding  ideal lattice in $\mathbb{Z}[x]/f(x)$, then we  have
$$\sigma(x\sigma^{-1}(\mathbf{h}_n)\mod f(x))\in \mathcal{L}(\mathbf{B}).$$
Note that 
$$x\sigma^{-1}(\mathbf{h}_n)\mod f(x) = \sum_{i=1}^{n-1} h_{n,i}x^i - h_{n,n} (f(x)-x^n).$$
We have
 \begin{equation} \label{equ1}
\begin{pmatrix} 0&h_{n,1}&...&h_{n,n-1} \end{pmatrix}-h_{n,n}\sigma(f(x)-x^n) \in  \mathcal{L}(\mathbf{B}).
 \end{equation}

 Similarly, since $\mathcal{L}(\mathbf{B})$ is also a coefficient-embedding  ideal lattice in $\mathbb{Z}[x]/g(x)$, we have 
  \begin{equation} \label{equ2}
 \begin{pmatrix} 0&h_{n,1}&...&h_{n,n-1} \end{pmatrix}-h_{n,n}\sigma(g(x)-x^n) \in  \mathcal{L}(\mathbf{B}).
  \end{equation}
 Subtracting the left side of (\ref{equ1}) from the left side of (\ref{equ2}), we immediately have
 $$h_{n,n}\sigma(f(x)-g(x)) \in \mathcal{L}(\mathbf{B}).$$
 By Lemma \ref{ihnflemma}, $h_{n,n} = d$, we have
 $$\sigma(f(x)-g(x)) \in \mathcal{L}(\frac{\mathbf{B}}{d}).$$

 (ii) We next prove that for any  polynomial $g(x)$, such that $\sigma(f(x)-g(x)) \in  \mathcal{L}(\frac{\mathbf{B}}{d})$, any  full-rank coefficient-embedding ideal lattice $\mathcal{L}(\mathbf{B})$ in the ring $\mathbb{Z}[x]/f(x)$ can also be viewed as a coefficient-embedding  ideal lattice in $\mathbb{Z}[x]/g(x)$.

 First, $g(x)$ is obviously a monic polynomial with degree $n$. Note that by Lemma \ref{klemma}, $h_{n,n}|h_{i,j}$, then $d = h_{n,n}$ divide all the components of every lattice vector in $\mathcal{L}(\mathbf{B})$, which means that $ \mathcal{L}(\frac{\mathbf{B}}{d})$ is an integer lattice and once $\sigma(f(x)-g(x)) \in  \mathcal{L}(\frac{\mathbf{B}}{d})$, $g(x)\in\mathbb{Z}[x]$.

 By Lemma \ref{klemma4ideal} again, 
 $\mathcal{L}(\mathbf{H}) = \mathcal{L}(\mathbf{B})$ is  a coefficient-embedding  ideal lattice in $\mathbb{Z}[x]/f(x)$, then we  have
 $$\sigma(x\sigma^{-1}(\mathbf{h}_i)\mod f(x))\in \mathcal{L}(\mathbf{B}),$$
 for $i = 1,\cdots, n$.
 
 To prove that  $\mathcal{L}(\mathbf{B})$ is also a coefficient-embedding  ideal lattice in $\mathbb{Z}[x]/g(x)$, by Lemma \ref{klemma4ideal}  it is enough to show that 
  $\sigma(x\sigma^{-1}(\mathbf{h}_i)\mod g(x))\in \mathcal{L}(\mathbf{B}),$
 for $i = 1,\cdots, n$.

 Note that for $i = 1,\cdots, n-1$, 
 $$\sigma(x\sigma^{-1}(\mathbf{h}_i)\mod g(x)) = \sigma(x\sigma^{-1}(\mathbf{h}_i)\mod f(x))\in \mathcal{L}(\mathbf{B}).$$
 
 Since$\sigma(f(x)-g(x)) \in  \mathcal{L}(\frac{\mathbf{B}}{d})$, there exists a lattice vector $\mathbf{v} \in \mathcal{L}(\mathbf{B})$ such that $d(f(x)-g(x)) = h_{n,n}(f(x)-g(x))= \sigma^{-1}(\mathbf{v})$. Then for $i = n$, 
 \begin{align*}
 \sigma(x\sigma^{-1}(\mathbf{h}_n)\mod g(x)) &=\sigma (\sum_{i=1}^{n-1} h_{n,i}x^i - h_{n,n} (g(x)-x^n))\\& = \sigma (\sum_{i=1}^{n-1} h_{n,i}x^i- h_{n,n} (f(x)-x^n) + \sigma^{-1}(\mathbf{v}))\\
 &= \sigma(x\sigma^{-1}(\mathbf{h}_n)\mod f(x)) + \mathbf{v} \in \mathcal{L}(\mathbf{B}).
 \end{align*}
 The  theorem  follows.
\end{proof}

\begin{remark}
	The HNF \ $\mathbf{H}$ in the proof can be replaced by any Incomplete Hermite Normal Form.
	
\end{remark}
\begin{remark}\label{rmk:app}
  For most lattice-based cryptosystems, their security is guaranteed by the hardness of lattice problems such as $\gamma$-SVP.  Hence, the hardness of lattice problem in ideal lattice is widely considered as the security foundation of Ring-LWE based cryptosystems. 

  However, the worst-case hardness of ideal lattice $\gamma$-SVP in different polynomial rings are not the same exactly. For example, in the ring $\mathbb{Z}[x]/(x^n+1)$ $n=2^k$ $k\geq 1$, there is a quantum polynomial time algorithm for ideal lattice $\text{exp}(n^{1/2})$-SVP \cite{ref7} \cite{ref8}, but the coefficient is no less than $\text{exp}(n)$ in the majority of polynomial rings.

  Theorem \ref{kthm} indicates that an ideal lattice can be viewed as an ideal lattice in infinitely different polynomial rings. Hence, it's possible to embed the given ideal lattice into a special ring such as $\mathbb{Z}[x]/(x^n+1)$ $n=2^k$ $k\geq 1$ which can help to the solve the hard lattice problems.
	
\end{remark}

\section{Identifying an Ideal Lattice}
\label{section:Identifying an Ideal Lattice}
\subsection{Algorithm}

According to Theorem \ref{thmidentify} and Theorem \ref{kthm},  we propose an algorithm to identify whether a given integer lattice is an ideal lattice or not (Algorithm 1).

\begin{algorithm}[htb]

	\caption{Identifying an ideal lattice}
	\label{alg:identify}
	\begin{algorithmic}[1]
		\Require $\mathbf{B} \in \mathbb{Z}^{n \times n}$, $\text{rank}(\mathbf{B})=n$.
		\Ensure False if $\mathcal{L}(\mathbf{B})$ is not a coefficient-embedding ideal lattice; Otherwise output a set $S\subset \mathbb{Z}^n$  s.t. for any $(g_1,g_2,...,g_n)\in S$, $\mathcal{L}(\mathbf{B})$ can be embedded as an ideal  into $\mathbb{Z}[x]/(g_1+g_2x^1+...+g_nx^{n-1}+x^n)$.
		\State  Compute  any  Incomplete Hermit Normal Form $\mathbf{B}'= \begin{pmatrix} \mathbf{D}&\mathbf{0} \\ \mathbf{b'}&b_{n,n}' \end{pmatrix}$ of $\mathbf{B}$ by unimodular transformation;
		\If{$b_{n,n}' \not\vert\  \mathbf{B}$} return False;
		\EndIf
	    \If {$\begin{pmatrix} \mathbf{0}&\mathbf{D}\end{pmatrix}\mathbf{B}^{-1}\notin\mathbb{Z}^{(n-1)\times n} $} return False;
		\EndIf
		\State  Output $S = \frac{1}{b'_{n,n}}(\begin{pmatrix}0&\mathbf{b}'\end{pmatrix}+\mathcal{L}(\mathbf{B}))$.
	\end{algorithmic}
\end{algorithm}

\begin{remark}\label{HNFremark}
	In Step 1, we can also compute the HNF of $\mathcal{L}(\mathbf{B})$, and then use the divisibility relation described in Lemma \ref{klemma} to rule out some integer lattices that can't be embedded as an ideal into any polynomial ring. This may speedup the algorithm in practice, since many "random" integer lattices can not pass such check.
\end{remark}

The correctness of Algorithm 1 is guaranteed by Theorem \ref{thmidentify} and Theorem \ref{kthm}

\subsection{Complexity}
For Step 1, we can use Algorithm 2 to compute an Incomplete Hermite Normal Form for   $\mathbf{B} \in \mathbb{Z}^{n \times n}$ with a unimodular transformation, whose idea has already been described in Section 2.3.

\begin{algorithm}[htb]
	\caption{Computing an Incomplete Hermite Normal Form}
	\label{alg:ihnf}
	\begin{algorithmic}[1]
		\Require $\mathbf{B} \in \mathbb{Z}^{n \times n}$, $\text{rank}(\mathbf{B})=n$.
		\Ensure An Incomplete Hermit Normal Form of $\mathbf{B}$ by unimodular transformation.
		\For{$i$ from 1 to $n-1$}
		\State Use Extended Euclidean Algorithm with input $(b_{i,n},b_{i+1,n})$ to find $x$, $y$, $d$ s.t. $xb_{i,n}+yb_{i+1,n}=\gcd(b_{i,n},b_{i+1,n})=d$;
		\State  Update	$\begin{pmatrix} \mathbf{b_i} \\ \mathbf{b_{i+1}}\end{pmatrix}$:= $\begin{pmatrix}-b_{i+1,n}/d&b_{i,n}/d\\x&y\end{pmatrix}\begin{pmatrix} \mathbf{b_i} \\ \mathbf{b_{i+1}}\end{pmatrix}$;
		\EndFor
		\State  Output $\mathbf{B}$.
	\end{algorithmic}
\end{algorithm}

It is easy to check that the integer matrix $\begin{pmatrix}-b_{i+1,n}/d&b_{i,n}/d\\x&y\end{pmatrix}$ is unimodular since its determinant is $-1$. Hence, the transformation in Step 3 will not change the lattice $\mathcal{L}(\mathbf{B})$. After Step 3 for each $i$, we have $b_{i,n} = 0$ and $b_{i+1,n}= d$ computed by Step 2, which means that the output is in Incomplete Hermite Normal Form.

For the time complexity, we assume that for the input  $\mathbf{B}$, the absolute value of its every entry is bounded by $2^B$. 

It is easy to conclude that for the $i$-th loop, at the beginning, we have 
\begin{itemize}
	\item $|b_{i,j}|< 2^{i*B+1}$, $|b_{i+1,j}|< 2^{B}$ for $j =1,\cdots, n$, especially we have $|b_{i,n}|< 2^{B}$;
	\item $|x|<  2^{B}$, $|y|< 2^{B}$, $d<  2^{B}$.
\end{itemize}
Note that the Extended Euclidean Algorithm takes $\mathcal{O}(\text{log}\vert a \vert \text{log}\vert b\vert)$ bit operations on input $(a, b)$. Then for the $i$-th loop, with the plain integer multiplication we have:
\begin{itemize}
	\item Step 2 costs $\mathcal{O}(B^2)$ bit operations;
	\item Step 3 costs $\mathcal{O}(i*nB^2)$ bit operations;
\end{itemize} 
Hence, for the total $n$ loops, Algorithm \ref{alg:ihnf} needs $\mathcal{O}(n^3B^2)$ bit operations, and we have the following result.

\begin{lemma}\label{clemma}
	For a non-singular matrix $\mathbf{B}\in\mathbb{Z}^{n \times n}$, the absolute value of whose entries is bounded by $2^B$, Algorithm \ref{alg:ihnf} takes $\mathcal{O}(n^3B^2)$ bit operations to compute an Incomplete Hermite Normal Form of $\mathbf{B}$ by a unimodular transformation.
\end{lemma}

The most time-consuming part of Algorithm 1 is to judge whether $\begin{pmatrix}0&\mathbf{D}\end{pmatrix}\mathbf{B}^{-1}\in\mathbb{Z}^{(n-1)\times n}$ or not. In fact, there is an equivalent description for this and we refer to the results of Birmpilis et al \cite{ref30}.

\begin{theorem}[See Theorem 4 of \cite{ref30} ]\label{LatMemthm}
Let $\mathbf{B}\in\mathbb{Z}^{n\times n}$ be nonsingular with Smith form $\mathbf{S}$ and Smith massager $\mathbf{M}$. Let $s$ be the largest invariant factor of $\mathbf{S}$. The following lattices are identical:

$L_1=\{v|v\mathbf{B}^{-1}\in\mathbb{Z}^{1\times n}\}$

$L_2=\{v|v\mathbf{M}\equiv0_{1\times n}$ cmod $\mathbf{S}\}$
\end{theorem}

 By Theorem \ref{LatMemthm}, $L_1=L_2$, which means to judge whether $\begin{pmatrix}0&\mathbf{D}\end{pmatrix}\mathbf{B}^{-1}$ $\in\mathbb{Z}^{(n-1)\times n}$ or not , it's sufficient to verify $\begin{pmatrix}0&\mathbf{D}\end{pmatrix}\mathbf{M}\equiv 0_{(n-1)\times n}$ cmod $\mathbf{S}$. $\mathbf{S}$ is the Smith Norm Form of $\mathbf{B}$, and it's diagonal.  
 The following theorem is also proposed by Birmpilis et al
\cite{ref30} to compute the Smith Normal Form $\mathbf{S}$ and a reduced Smith Massager $\mathbf{M}$ of the input matrix ($\mathbf{M}$ is reduced $\mathbf{c}$mod $\mathbf{S}$)

\begin{theorem}[See Theorem 19 of \cite{ref30}]\label{Sthm}
There exists a Las Vegas algorithm that takes as input a nonsingular $\mathbf{A}\in\mathbb{Z}^{n\times n}$, and returns as output the Smith Normal Form $\mathbf{S}\in\mathbb{Z}^{n\times n}$ and a reduced Smith Massager $\mathbf{M}\in\mathbb{Z}^{n\times n}$ of the input matrix. The cost of the algorithm is $\mathcal{O}(n^{\omega}\text{B}(\log n+\log \|A\|)(\log n)^2)$ bit operations. The algorithm returns Fail with probability at most 1/2.
\end{theorem}

$\text{B}(d)=\mathcal{O}(M(d)\log d)$ and $M(d)$ bounds the number of bit operations required to multiply two integers bounded in magnitude by $2^d$.  We take $M(d)=\mathcal{O}(d^2)$. $\omega$ is a valid exponent of matrix multiplication: two $n\times n$ matrices can be multiplied in $\mathcal{O}(n^{\omega})$ operations from the domain of the entries, and the best known upper bound is $\omega<2.37286$ by Alman and Williams \cite{ref31}.

After computing $\mathbf{M}$ and $\mathbf{S}$, it remains to compute $\begin{pmatrix}0&\mathbf{D}\end{pmatrix}\mathbf{M}$ $\mathbf{c}$mod $\mathbf{S}$. It's easy to check that the entries of the output of Algorithm 2 $\begin{pmatrix}0&\mathbf{D}\end{pmatrix}$ are bounded by $2^{3B}$.  We first consider the computation of the remainder modulo $Y$ of the product of two integers. Recall that $\text{Rem}(ab,Y)$ has cost bounded by $\mathcal{O}((\log ab/Y)(\log Y))$, so the cost of computing the $i$-th element of  $\begin{pmatrix}0&\mathbf{D}\end{pmatrix}\mathbf{M}$ $\mathbf{c}$mod $\mathbf{S}$ is bounded by $\mathcal{O}(n(3B)(\log s_i))$. Hence, the cost of computing  $\begin{pmatrix}0&\mathbf{D}\end{pmatrix}\mathbf{M}$ $\mathbf{c}$mod $\mathbf{S}$ is bounded by $\mathcal{O}(3nB\log |\text{det}(\mathbf{B})|)=\mathcal{O}(3n^2B(B+\log n))$

Combining Lemma 7, Theorem 4 and the discussion above, we have

\begin{theorem}\label{complexity}
Given $\mathbf{B} \in \mathbb{Z}^{n \times n}$, $\text{rank}(\mathbf{B})=n$, and the absolute value of the entries of $\mathbf{B}$ is bounded by $2^B$, then there is a Las Vegas algorithm with expected complexity  $\mathcal{O}(n^3B(B+\log n))$ to identify whether $\mathcal{L}(\mathbf{B})$ is an ideal lattice or not.
\end{theorem}

\subsection{Related research}

In 2007, Ding and Lindner \cite{ref13} already proposed an algorithm for identifying ideal lattices, but we find that there is a flaw in their algorithm. More exactly, some ideal lattices can't be identified by their algorithm.

We find some non-trivial ideal lattices which can't be identified by Ding and Lindner's algorithm. The following is an example:
\[\mathbf{B}=\begin{pmatrix}6&-8&-5\\3&-7&-4\\6&1&-1\end{pmatrix}\]

The row vectors of $\mathbf{B}$ span a full-rank ideal lattice in the ring $\mathbb{Z}[x]/x^3+3x^2+x^1-3$. However, with the input $\mathbf{B}$, Ding and Lindner's algorithm return false. 

More exactly, in their algorithm, the lattice is spanned by column vectors, so the input matrix should be $\mathbf{B}^T$. They first transform $\mathbf{B}^T$ into an upper-triangular Hermite Normal Form $\mathbf{H}$.
\[\mathbf{H}=\begin{pmatrix}9&6&0\\0&1&0\\0&0&1\end{pmatrix}\]
Then they compute the adjugate matrix $\mathbf{A}$ of $\mathbf{H}$.
\[\mathbf{A}=\begin{pmatrix}1&-6&0\\0&9&0\\0&0&9\end{pmatrix}\]
Let $\mathbf{I_n}$ be the unit matrix of dimension $n$, and $\mathbf{M}$ be a matrix only related to the dimension $n$ (For this example, the dimension is 3).
\[\mathbf{M}=\begin{pmatrix}\mathbf{0}&0\\\mathbf{I_{n-1}}&\mathbf{0}\end{pmatrix}\]
In step 4 of their algorithm, they need to verify whether only the last column $\mathbf{AMH}$ mod $\det(\mathbf{B})$ is equal to $\mathbf{0}$ or not. If the input lattice basis $\mathbf{B}$ spans an ideal lattice, they believe by default only the last column $\mathbf{AMH}$ mod $\det(\mathbf{B})$ is not equal to $\mathbf{0}$. However, $\mathbf{AMH} \equiv \mathbf{0}$ mod $\det(\mathbf{B})$, which causes their algorithm to return "false". Apparently, they ignore the situation that all the column of $\mathbf{AMH}$ mod $\det(\mathbf{B})$ is equal to $\mathbf{0}$.

Ignoring the flaw above, our algorithm still performs better than theirs in two aspects:
\begin{itemize}
    \item Our algorithm outputs more. Ding and Lindner's algorithm outputs a single polynomial ring of the ring class if the input lattice is an ideal lattice but ours outputs the entire ring class.
    \item The time complexity of our algorithm is lower. It is claimed  in \cite{ref13}  that the algorithm presented by Ding and Lindner to identify an ideal lattice  costs $\mathcal{O}(n^4B^2)$ bit operations. However, we have to point out that there is also a flaw in the complexity analysis in  $\mathcal{O}(n^4B^2)$.  The algorithm in \cite{ref13}  needs to compute $n-2$ powers of $\mathbf{B}$, that is, $\mathbf{B}^k$ for $k = 2, \cdots, n-1$. It is claimed this can be done within $\mathcal{O}(n^4B^2)$ bit operations. However, when $k$ grows bigger, the bit size of the entries in $\mathbf{B}^k$ will be $\mathcal{O}(kB)$ instead of $B$. Hence the correct time complexity  should be 
$$\sum_{k=2}^{n-1}\mathcal{O}(n^3*k*B^2) =  \mathcal{O}(n^5B^2).$$
\end{itemize}



\subsection{Experiment}
Using our algorithm, we conducted several experiments, and the experimental results are presented in Appendix~\ref{AppendixA}.
\section{Conclusion}
In this paper, we explore the connection between integer lattices and coefficient-embedding ideal lattices. We have three main contributions:

Firstly, we find and proof an ideal lattice can be viewed as an ideal lattice in infinitely many different polynomial rings. This interesting phenomenon may contribute to the solution to hard ideal lattice problems as mentioned in Remark \ref{rmk:app}.

Secondly, we propose an efficient algorithm for identifying ideal lattices, and compared to related work, our algorithm has more advantages.

Finally, we provide an efficient open source implementation of our algorithm for identifying ideal lattices in SageMath. Our experimental results are presented in Appendix~\ref{AppendixA}.

 


\newpage
\bibliographystyle{unsrt} 
\bibliography{sample-base}

\appendix
\section{Experiments}
\label{AppendixA}
In this section, we present our experimental results and some intersting findings about density of ideal lattice. The experiments were conducted in the SageMath 9 environment on a personal computer equipped with an Intel Core i7-13700KF 3.40 GHz processor. The source code for the experiments is open-sourced and available at \begin{center}
  \textcolor{blue}{\url{https://github.com/fffmath/Identifying-Ideal-Lattice}}.
\end{center} It allows simulations of experiments with input dimensions, bounds, and the number of experiments.

We compared our algorithm with the one proposed by Ding and Lindner~\cite{ref13}. Under the same parameters, our algorithm demonstrated a significant advantage in terms of runtime.


Regarding algorithm runtime, we conducted multiple experiments with different variables. For input parameters \textit{dim} and \textit{bound}, we randomly generated a \textit{dim}-dimensional matrix within the specified \textit{bound} as the lattice basis. In other words, this results in the generation of a dim $\times$ dim matrix, where each element of the matrix falls within the range of $-2^{\textit{bound}}$ to $2^{\textit{bound}}$.

Two scenarios were considered:

\begin{itemize}
    \item Fixing the dimension (\textit{dim}): We kept \textit{dim} constant and recorded the runtime as \textit{bound} increased gradually.
    \item Fixing the bound (\textit{bound}): We kept \textit{bound} constant and recorded the runtime as \textit{dim} increased.
\end{itemize}

The relevant experimental results can be found in Figure~\ref{fig:lattice}.

For parameters with \textit{dim} less than 300, we conducted 100 experiments for each parameter and recorded the average time consumption as the time record. We observed that these data have very low variance, with each data point closely approaching the mean.

For parameters with large \textit{dim}, due to the longer individual runtime, we performed five experiments for each group and used the average of these five values as the time consumption.

\begin{figure}[htbp]
    \hspace{-2em} 
    \subfloat[dim fixed]
    {
        \begin{minipage}{0.25\textwidth} 
            \centering
            \includegraphics[scale=0.23]{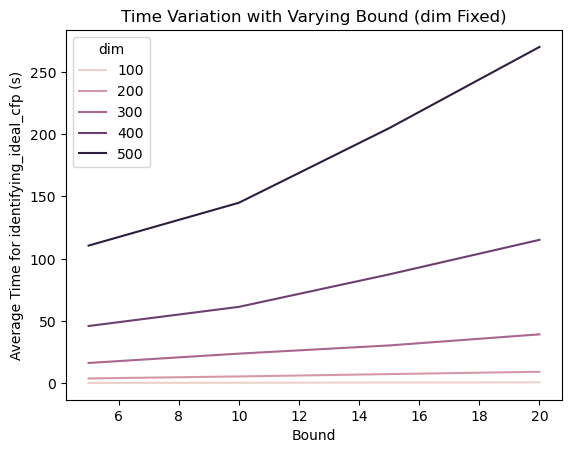}
        \end{minipage}
    }
    \subfloat[bound fixed]
    {
        \begin{minipage}{0.25\textwidth} 
            \centering
            \includegraphics[scale=0.23]{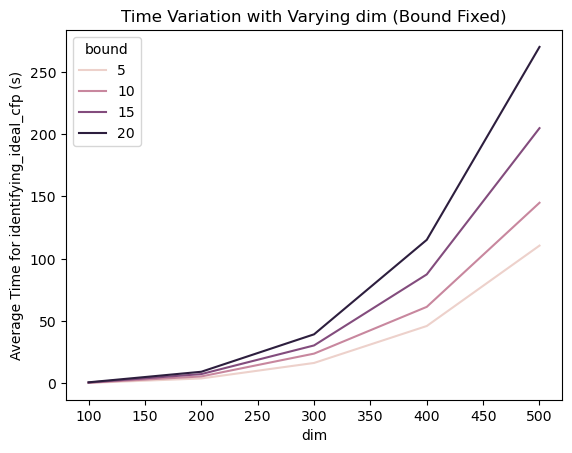}
        \end{minipage}
    }
    \caption{Cost time for our algorithm using random lattice as input} %
    \label{fig:lattice}
\end{figure}

Note that as dimensions or bounds increased, the proportion of ideal lattices became very small. Therefore, most of the generated lattices in the former experiments weren't ideal lattice, resulting in runtime data just be not suitable for ideal lattice input.

To further explore ideal lattices, we conducted additional experiments using ideal lattice as input. We randomly selected polynomials $f$ with coefficients in \{-1,0,1\} and $g$ with coefficients in (${-2^{\textit{bound}}, 2^{\textit{bound}}}$) and computed the lattice basis of the principal ideal generated by $g$ in $\mathbb{Z}[x]/f(x)$, ensuring it is an ideal lattice. In such case, we take the coefficient vectors of $x^ig(x)\text{mod}f(x)$ as the lattice basis, and the reason why we limit the coefficients of $f(x)$ in \{-1,0,1\} is to decrease the exploration of the coefficients of ideal lattice basis generated by $g$. Similarly as former experiments, we also performed experiments with fixed dimensions, recording the runtime as \textit{bound} varied, and fixed bounds, recording the runtime as \textit{dim} varied. The relevant experimental results can be found in Figure~\ref{fig:ideal-lattice}. 

\begin{figure}[htbp]
    \hspace{-2em} 
    \subfloat[dim fixed]
    {
        \begin{minipage}{0.25\textwidth} 
            \centering
            \includegraphics[scale=0.23]{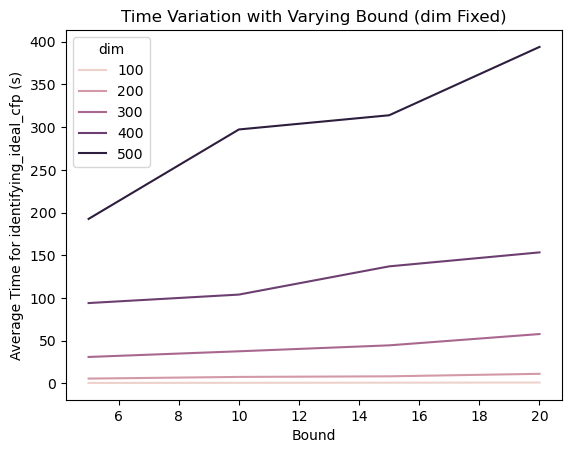}
        \end{minipage}
    }
    \subfloat[bound fixed]
    {
        \begin{minipage}{0.25\textwidth} 
            \centering
            \includegraphics[scale=0.23]{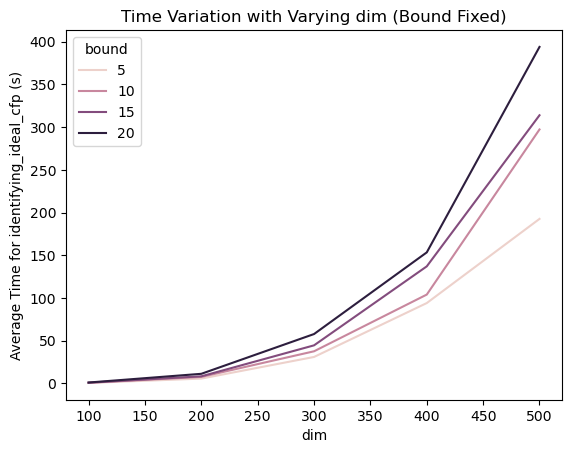}
        \end{minipage}
    }
    \caption{Cost time for our algorithm using random lattice as input} %
    \label{fig:ideal-lattice}
\end{figure}

To facilitate the comparison of different parameters and the runtime under various inputs, you can refer to the data table in Table~\ref{table:experiments}.
\begin{table}[htbp]
    \centering
\begin{tabular}{ccc}
   \toprule
   (dim, bound) & lattice (s)& ideal lattice (s)\\
   \midrule
    (100, 5) & 0.406 & 0.467 \\
    (100, 10) & 0.555 & 0.598 \\
    (100, 15) & 0.713 & 0.759 \\
    (100, 20) & 0.894 & 0.934 \\
    (200, 5) & 3.999 & 5.538 \\
    (200, 10) & 5.607 & 7.503 \\
    (200, 15) & 7.494 & 8.203 \\
    (200, 20) & 9.365 & 11.140 \\
    (300, 5) & 16.426 & 30.870 \\
    (300, 10) & 23.916 & 37.507 \\
    (300, 15) & 30.485 & 44.475 \\
    (300, 20) & 39.398 & 57.703 \\
    (400, 5) & 46.075 & 93.985 \\
    (400, 10) & 61.436 & 103.909 \\
    (400, 15) & 87.487 & 136.954 \\
    (400, 20) & 115.221 & 153.318 \\
    (500, 5) & 110.583 & 192.532 \\
    (500, 10) & 144.965 & 297.249 \\
    (500, 15) & 204.832 & 313.888 \\
    (500, 20) & 270.002 & 393.900 \\
   \bottomrule
\end{tabular}
\caption{Experimental results for cost time when using random lattice/ideal lattice as input.}
\label{table:experiments}
\end{table}

Finally, although finding an ideal lattice in high dimensions is challenging, we conducted experiments in lower dimensions to estimate the reduction factor. We investigated the density of ideal lattices in low dimensions and small bounds. We performed 100,000 experiments for each parameter $\textit{dim}=3$, $\textit{bound}=3,4,5,6,7$ and $\textit{bound}=3$, $\textit{dim}=2,3,4,5,6$, recording the quantity of ideal lattices under different parameters. We observed a rapid decrease in the proportion of ideal lattices in Figure~\ref{fig:density}.
\begin{figure}
    \hspace{-2em} 
	\subfloat[bound fixed]
	{
		\begin{minipage}{0.25\textwidth}
			\centering
			\includegraphics[scale=0.23]{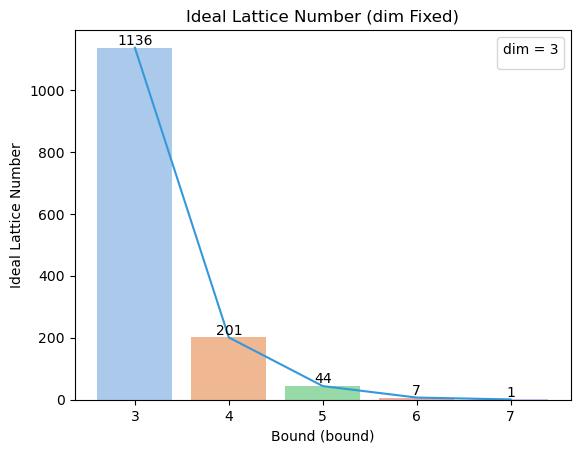}
		\end{minipage}
	}
	\subfloat[dim fixed]
	{
		\begin{minipage}{0.25\textwidth}
			\centering
			\includegraphics[scale=0.23]{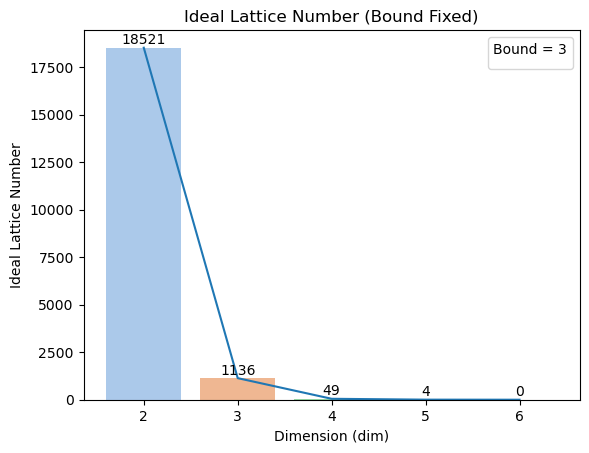}
		\end{minipage}
	}
	\caption{Density of ideal lattice} %
 \label{fig:density}
\end{figure}

\end{document}